\documentclass[12pt]{article}
\usepackage{amsmath,amsfonts}
\usepackage{amsthm}
\usepackage{amssymb}
\usepackage{algorithm}
\usepackage[noend]{algorithmic}

\newtheorem{theorem}{Theorem}
\newtheorem{lemma}{Lemma}

\newtheorem{claim}{Claim}

\def\Real{\mathbb R}
\def\Natural{\mathbb N}
\def\cover{\text{Cover}}
\def\xcluster{\text{XC}}
\def\ycluster{\text{YC}}
\def\Linfty{l_{\infty}}
\def\disk{\delta}

\def\primary{\overline{Y}}
\def\largest{\text{largest}}
\def\oc{\text{OuterCover}}

\usepackage{paralist}%
\usepackage[cm]{fullpage}%
\usepackage{color}%
\usepackage{xspace}%

\usepackage[pdfpagelabels,plainpages=false,hypertexnames=false]{hyperref}%
\hypersetup{%
   breaklinks,%
   ocgcolorlinks,
   colorlinks=true,%
   linkcolor=[rgb]{0.45,0.0,0.0},%
   urlcolor=[rgb]{0.05,0.390,0.0},%
   citecolor=[rgb]{0,0,0.45}
}%

\usepackage{pifont}%
\usepackage[symbol*]{footmisc}

\newcommand*\circled[1]{\footnotesize\tikz[baseline=(char.base)]{
    \node[shape=circle,draw,inner sep=0.2pt] (char) {#1};}}

\DefineFNsymbols*{stars}[text]{%
   {\ding{172}} %
   {\ding{173}} %
   {\ding{174}} %
   {\ding{175}} %
   {\ding{176}} %
   {\ding{177}} %
   {\ding{178}} %
   {\ding{179}} %
   {\ding{180}} %
   {{\protect\circled{{\tiny 10}}}}%
   {{\protect\circled{{\tiny 11}}}}%
   {{\protect\circled{{\tiny 12}}}}%
   {{\protect\circled{{\tiny 13}}}}%
   {{\protect\circled{{\tiny 14}}}}%
   {{\protect\circled{{\tiny 15}}}}%
   {{\protect\circled{{\tiny 16}}}}%
   {{\protect\circled{{\tiny 17}}}}%
   {{\protect\circled{{\tiny 18}}}}%
   {{\protect\circled{{\tiny 19}}}}%
   {{\protect\circled{{\tiny 20}}}}%
   {{\protect\circled{{\tiny 21}}}}%
   {{\protect\circled{{\tiny 22}}}}%
}%

\begin{document}
\title{A Constant-Factor Approximation for Multi-Covering with Disks\thanks{A preliminary
   version of this article appeared as \cite{Bhowmick2013Constant}.}
}
  
\author{%
  Santanu~Bhowmick%
  \thanks{Department of Computer Science; University of Iowa; Iowa City, USA; 
          \url{santanu-bhowmick@uiowa.edu}}\,
  \and%
  Kasturi~Varadarajan%
  \thanks{Department of Computer Science; University of Iowa; Iowa City, USA;
          \url{kasturi-varadarajan@uiowa.edu}}\,
  \and%
  Shi-Ke~Xue%
  \thanks{Department of Computer Science; Massachusetts Institute of Technology; Cambridge, USA;
          \url{shikexue@mit.edu}}
}%

\maketitle
\setfnsymbol{stars}

\begin{abstract}
 We consider variants of the following multi-covering problem with disks. We are
given two point sets $Y$ (servers) and $X$ (clients) in the plane, a coverage 
function $\kappa :X \rightarrow \Natural$, and a constant $\alpha \geq 1$. Centered
at each server is a single disk whose radius we are free to set. The requirement
is that each client $x \in X$ be covered by at least $\kappa(x)$ of the server disks.
The objective function we wish to minimize is the sum of the $\alpha$-th powers of the disk radii. We 
present a polynomial time algorithm for this problem achieving an $O(1)$ approximation. 
\end{abstract}

\section{Introduction}
We begin with the statement of the problem studied in this article.
We are given two point sets $Y$ (servers) and $X$ (clients) in the plane,
a coverage function $\kappa: X \rightarrow \Natural$, and a constant $\alpha \geq 1$. 
An assignment $r: Y \rightarrow \Real^+$ of {\em radii} to the points in $Y$ 
corresponds to ``building'' a disk of radius $r_y$ centered at each $y \in Y$.
For an integer $j \geq 0$, let us say that a point $x \in X$ is $j$-covered 
under the assignment if $x$ is contained in at least $j$ of the disks, i.e. 
\[| \{ y \in Y \ | \ ||y-x||_2 \leq r_y \}| \geq j\] 
The goal is to find an assignment that $\kappa(x)$-covers each point $x \in X$ 
and minimizes $\sum_{y \in Y}r_y^{\alpha}$. We call this the {\em non-uniform
minimum-cost multi cover} problem (non-uniform MCMC problem).

We are interested in designing a polynomial time 
algorithm that outputs a solution whose cost is at 
most some factor $f \geq 1$ times the cost of an optimal solution. We call such 
an algorithm an $f$-approximation, and it is implicit that the algorithm is actually polynomial-time. 

The version of this problem where $\kappa(x) = k, \ \forall x \in X$, for some
given $k > 0$, has received particular attention. Here, all the clients have the same coverage requirement of $k$. We will refer to this as the {\em uniform MCMC} problem. In the context of the uniform MCMC, we will refer to a $j$-cover
as an assignment of radii to the servers under which each client is $j$-covered. 

\subsection{Related Work}
In the rest of this section, we will focus on the uniform MCMC problem, and be specific when remarking on generalizations to the non-uniform problem. 
The (uniform) MCMC problem was considered in two recent papers, motivated by fault-tolerant
sensor network design that optimizes energy consumption.  Abu-Affash et al.~\cite{ACKM} considered the case $\alpha = 2$, which corresponds to minimizing the sum of the areas of the server disks. They gave an $O(k)$ approximation for the problem using mainly geometric ideas. Bar-Yehuda 
and Rawitz \cite{Yehuda2013Note} gave another algorithm that achieves the same approximation 
factor of $O(k)$ for any $\alpha$, using an analysis based on the local ratio technique. The central question that we investigate in this article is whether an approximation guarantee that is independent of $k$ is possible. 

There is a considerable amount of work on clustering and covering problems related
to the MCMC problem, and we refer the reader to the previous papers for a detailed 
survey \cite{ACKM, Yehuda2013Note}. Here, we offer a view of some of that work from 
the standpoint of techniques that may be applicable to the problem at hand. For the 
case $k = 1$ of the problem, constant factor approximations can be obtained using 
approaches based on linear programming, and in particular, the primal-dual method 
\cite{Charikar2001Clustering,Freund2004Combinatorial}. The $O(k)$ approximation of 
Bar-Yehuda and Rawitz \cite{Yehuda2013Note} for $k > 1$ can be situated in this line 
of work.

There has been some recent work on the geometric set multi-covering problem 
\cite{Chekuri2009Set,Bansal2012Weighted}. In particular, the recent work of Bansal 
and Pruhs \cite{Bansal2012Weighted} addresses the following problem. We are a given 
a set of points in the plane, a set of disks each with an arbitrary non-negative 
weight, and an integer $k$. The goal is to pick a subset of the disks so that each 
of the given points is covered at least $k$ times. The objective function we want 
to minimize is the sum of the weights of the chosen disks. Bansal and Pruhs 
\cite{Bansal2012Weighted} give an $O(1)$ approximation for the problem, building 
on techniques developed for the case $k = 1$ \cite{Varadarajan2010Weighted, Chan2012Weighted}. 

It would seem that the problem considered in this paper can be reduced to the problem 
solved by Bansal and Pruhs: for each $y \in Y$ and $x \in X$, add a disk centered at 
$y$ with radius $||x-y||_2^{\alpha}$, and let $X$ be the set of points that need to be covered. 
The reason this reduction does not work is that we have to add an additional constraint 
saying that we can use only one disk centered at each $y \in Y$. Notice that this 
additional constraint is not an issue for the case $k = 1$, since here if the returned 
solution uses two disks centered at the same $y \in Y$, we can simply discard the smaller one.

In the geometric set cover problems considered by \cite{Chekuri2009Set, 
Varadarajan2010Weighted, Chan2012Weighted, Bansal2012Weighted}, the
input disks are ``immutable'', and the complexity of the problem stems from the
combinatorial geometry of the disks. For the MCMC application, it would be
more fruitful to consider geometric set multi-cover problems where the algorithm
is allowed to slightly enlarge the input disks. This version of covering with
$k = 1$ is considered by Har-Peled and Lee \cite{HarPeled2012Weighted}. For $k > 1$,
however, we still have the above-mentioned difficulty of reducing MCMC to a set multi-cover problem.

The case $k = 1$ of our MCMC problem actually admits a polynomial time 
approximation scheme (PTAS) using dynamic programming on top of randomly 
shifted quad-trees \cite{Erlebach2005PolynomialTime,Chan2003Polynomialtime}. 
This was shown by the work of Bilo et al. \cite{Bilo2005Geometric}, following 
the work of Lev-Tov and Peleg \cite{Tov2005Polynomial} for $\alpha = 1$. 
The difficulty with extending these results for $k = 1$ to general $k$ is that 
the ``density'' of the solution grows with $k$, and therefore the number of 
sub-problems that the dynamic program needs to solve becomes exponential in $k$. 
It is conceivable that further discretization tricks \cite{HarPeled2012Weighted}  
can be employed to get around this difficulty, but we have not succeeded in this 
effort. On the other hand, we are also not aware of any hardness result that rules 
out a PTAS. The problem is known to be NP-hard even for $k = 1$
and any $\alpha > 1$ \cite{Bilo2005Geometric, Alt2006}.

\subsection{Our Results}

In this article, we obtain an $O(1)$ approximation for the uniform MCMC problem. That is, we demonstrate an approximation bound that is independent of $k$. 

Our approach revolves around the notion of an {\em outer cover}. This is an
assignment of radii to the servers under which each client $x \in X$ is
covered by a disk of radius at least $||y^k(x)-x||_2$, where $y^k(x)$ is the $k$-th nearest neighbor of $x$ in $Y$. To motivate the notion, consider any
$k$-cover, and in particular, the optimal one. Consider the set of disks
obtained by picking, for each client $x \in X$, the largest disk covering $x$
in the $k$-cover. (Several clients can contribute the same disk.) This set
of disks is seen to be an outer cover.

We provide a mechanism for extending any $(k-1)$-cover to a 
$k$-cover so that the increase in objective function cost is bounded by
a constant times the cost of an optimal outer cover. This naturally
leads to our algorithm in Section~\ref{sec:algo} -- recursively compute a $(k - 1)$-cover
and then extend it to a $k$-cover. To bound its approximation ratio, we
argue in Section \ref{sec:analysis} that the optimal solution can be partitioned into a $(k-1)$-cover
and another set of disks that is almost an outer cover. Finally, we need
a module for computing an approximately optimal outer cover. We show in
Section \ref{sec:outercover} that
an existing primal-dual algorithm for $1$-covering can be generalized
for this purpose.

The idea of an outer cover has its origins in the notion of {\em primary disks} used by Abu-Affash et al.~\cite{ACKM}. Our work develops the idea and its 
significance much further, and this is partly what enables our $O(1)$ approximation bound.
  
Our algorithm and approximation guarantee of $O(1)$ works for the non-uniform MCMC problem as well. We therefore present our work in this slighly more general setting. 

\section{Preliminaries}
\label{sec:terminology}
For convenience, we solve the variant of the non-uniform MCMC problem where we have
$\Linfty$ disks rather than $l_2$ disks.
Our input is two point sets $Y$ and $X$ in $\Real^2$, a coverage function $\kappa: X \rightarrow \Natural
\cup \{0\}$, and the constant $\alpha \geq 1$. (It will be useful to allow $\kappa(x)$ to be $0$ for some $x \in X$.) We also assume that
$\kappa(x) \leq |Y|$ for each $x \in X$, for otherwise there is no feasible solution.

We describe 
an algorithm for assigning a {\em radius} $r_y \geq 0$ for each $y \in Y$,
with the guarantee that for each $x \in X$, there are at least $\kappa(x)$ points 
$y \in Y$ such that the $\Linfty$ disk of radius $r_y$ centered at $y$ contains $x$. In other words the guarantee is that for each $x \in X$,
\[ | \{ y \in Y \ | \ ||x-y||_{\infty} \leq r_y \} | \geq \kappa(x)\]

Our objective is to minimize $\sum_{y \in Y} r_y^{\alpha}$. For this optimization problem, we will show
that our algorithm outputs an $O(1)$ approximation. Clearly, this also
gives an $O(1)$ approximation for the original problem, where distances
are measured in the $l_2$ norm. We will use 
$||\cdot||$ to denote the $\Linfty$ norm. 

For each $x \in X$, fix an ordering of the points in $Y$ that is non-decreasing
in terms of $\Linfty$ distance to $x$. For $1 \leq j \leq |Y|$, let $y^j(x)$
denote the $j$-th point in this ordering. In other words, $y^j(x)$ is the
$j$-th closest point in $Y$ to $x$. For brevity, we denote $y^{\kappa(x)}(x)$ by $y^{\kappa}(x)$.

Let $\disk(p,r)$ denote the $\Linfty$ disk of radius $r$ centered at $p$.
The {\em cost} of a set of disks is defined to the sum of the $\alpha$-th
powers of the radii of the disks. The cost of an assignment of radii to the
servers is defined to be the cost of the corresponding set of disks.

\section{OuterCover: Algorithm to generate a preliminary cover}
\label{sec:outercover}
Given $X' \subseteq X$, $Y$, $\kappa$ and $\alpha \geq 1$, an {\em outer cover}
is an assignment $\rho: Y \rightarrow \Real^+$
of radii to the servers such that
for each client $x \in X'$, there is a server $y \in Y$ such that
\begin{enumerate}
	\item The disk $\delta(y, \rho_y)$ contains $x$
	\item Disk radius $\rho_y \geq ||x - y^{\kappa}(x)||$
\end{enumerate}

Our goal in this section is to compute an outer cover that minimizes the cost $\sum_y \rho_y^{\alpha}$.
In the rest of this section, we describe and analyze a procedure
$\oc(X',Y,\kappa, \alpha)$ that returns an outer cover 
$\rho: Y \rightarrow \Real^+$ whose cost is $O(1)$ times that of an optimal
outer cover. Since this result is used as a black box in our algorithm for
the non-uniform MCMC, the remainder of this section could be skipped on
a first reading.

The procedure $\oc(X',Y,\kappa, \alpha)$ is implemented via a modification of 
the primal-dual algorithm of Charikar and Panigrahy~\cite{Charikar2001Clustering}.
Note that their algorithm can be viewed as solving the case where $\kappa(x)
= 1$ for each $x \in X'$. As we will see, their algorithm and analysis readily
generalize to the problem of computing an outer cover.

\subsection{Linear Programming Formulation}
We begin by formulating the problem of finding an optimal outer cover as 
an integer program. 
For each server $y_i \in Y$ and radius $r \geq 0$, let $z_i^{(r)}$ be an indicator variable that denotes whether 
the disk $\disk(y_i,r)$ is chosen in the outer cover.\footnote{For a server
$y_i \in Y$, only the disks whose radius is from the
set $\{ ||y_i - x_j|| \ \mid \ x_j \in X' \}$ will play a role in much of our
algorithm. For describing the algorithm, however, it will be convenient to
allow any $r \geq 0$.}   
For any server $y_i \in Y$ and client $x_j \in X'$, we define the
\textit{minimum eligible radius} $R_{\min}(y_i, x_j)$ to be:
\[
    R_{\min}(y_i, x_j) = \max (||y_i- x_j||, ||y^{\kappa}(x_j) - x_j||)
\]

A disk centered at $y_i$ serves $x_j$ in an outer cover exactly when its
radius is at least $R_{\min}(y_i, x_j)$. Finally, let $C_i(r) = \{ x_j \in X' \ \mid \ r \geq R_{\min}(y_i,x_j)\}$. The set $C_i(r)$ consists of those clients
that $\disk(y_i,r)$ can serve. 

The problem of computing an optimal outer cover is that of minimizing
\begin{equation}
   \sum_{i,r} r^{\alpha} \cdot z_i^{(r)},
\end{equation}
subject to the constraints
\begin{eqnarray}
\label{eq:CovConstraint}
   \sum_{i,r: x_j \in C_i(r)} z_i^{(r)} & \geq & 1, \ \forall x_j \in X'\\
\label{eq:IntConstraint}
z_i^{(r)} & \in & \{0, 1\}, \ \forall i,r.
\end{eqnarray}

The first constraint, equation (\ref{eq:CovConstraint}), represents the condition
that for every client $x_j \in X'$, at least one disk that is capable
of serving it is chosen. The second constraint, equation (\ref{eq:IntConstraint}),
models the fact that the indicator
variables $z_i^{(r)}$ can only take boolean values $\{0, 1\}$. By relaxing the indicator
variables to be simply non-negative, i.e.
\begin{equation}
\label{eq:LPConstraint}
z_i^{(r)} \geq 0, \ \forall i,r,
\end{equation}
we get a linear program (LP), which we call the primal LP for the problem.

The dual of the above LP has a variable $\beta_j$ corresponding to every client $x_j \in X'$. The dual LP seeks to maximize

\begin{equation}
   \sum_{x_j \in X'} \beta_j,
\end{equation}

subject to the constraints
\begin{eqnarray}
\label{eq:dual_constraint}
\sum_{x_j \in C_i(r)} \beta_j & \leq & r^{\alpha}, \ \forall y_i, r\\
   \beta_j & \geq & 0, \ \forall x_j \in X'
\end{eqnarray}

\subsection{A Primal Dual Algorithm}
The primal dual algorithm is motivated by the above linear program. The 
algorithm maintains a dual variable $\beta_j$ for each client $x_j$. This
variable will always be non-negative and satisfy the dual constraints
(\ref{eq:dual_constraint}). If at some point in the algorithm, the
dual constraint (\ref{eq:dual_constraint}) holds with equality for some
$y_i$ and $r$, the disk $\disk(y_i,r)$ is said to be {\em tight}. A client
$x_j$ is said to be tight if there is some tight disk $\disk(y_i,r)$ such
that $x_j \in C_i(r)$. (Note that $\beta_j$ is then part of the dual constraint
(\ref{eq:dual_constraint}) that holds with equality.)

Our algorithm, $\oc(X',Y,\kappa,\alpha)$, initializes each $\beta_j$ to $0$,
which clearly satisfies (\ref{eq:dual_constraint}). The goal of the while
loop in lines 1 and 2, which we refer to as the {\em covering phase} of the
algorithm, is to ensure that each client in $X'$ becomes tight, that is, covered by
some tight disk. It is easy to see that the covering phase achieves this. We
note in passing that since the $\beta_j$ are never decreased in the covering
phase, a client or disk that becomes tight at some point remains tight for 
the rest of the phase. 

 \begin{algorithm*}[hbt]
 \caption{$\oc(X',Y,\kappa,\alpha)$}
 \begin{algorithmic}[1]
 \WHILE {$\exists\ x_j \in X'$ that is not tight}
 \STATE Increase the non-tight variables $\beta_j$ arbitrarily till some constraint in (\ref{eq:dual_constraint}) becomes tight.
 \ENDWHILE
 \STATE Let $T$ be the set of tight disks.
 \STATE $F \leftarrow \varnothing$
 \WHILE {$T \neq \varnothing$}
        \STATE $\disk(y_i,r) \leftarrow$ The disk of largest radius in $T$
        \STATE $N \leftarrow$ Set of disks that intersect $\disk(y_i,r)$
        \STATE $F \leftarrow F \cup \{\disk(y_i,r)\}$
        \STATE $T \leftarrow T\ \backslash\ N$
 \ENDWHILE
 \STATE Assign $\rho: Y \rightarrow \Real^+$ as follows:
 \[
     \forall\ y_i \in Y, \rho(y_i) = 
     \begin{cases}
          3r, & \mbox{if } \disk(y_i,r) \in F \\
          0,  & \mbox{if } F \mbox{ contains no disk centered at } y_i 
     \end{cases}
 \]       
 \end{algorithmic}
\end{algorithm*} 

Steps 3--9 constitute the {\em coarsening phase} of the algorithm. This phase
starts with the set $T$ of tight disks computed by the covering phase. It computes a subset $F \subseteq T$ of pairwise disjoint disks by considering the disks 
in $T$ in non-increasing order of radii, and adding a disk to $F$ if it does
not intersect any previously added disk.

Step 10 constitutes the {\em enlargement phase}. Each disk in $F$ is expanded
by a factor of $3$, and the resulting set of disks is returned by the algorithm.
Note that for $y_i \in Y$, $F$ contains at most one disk centered at $y_i$; thus
the assignment in Step 10 is well defined.

We argue that the disks returned by $\oc(X', Y, \kappa, \alpha)$ form an
outer cover. Consider any client $x_j \in X'$. Since $x_j$ is tight at the end
of the covering phase, there is a tight disk $\disk(y_i,r)$ such that $x_j
\in C_i(r)$. Thus $x_j$ is served in case  $\disk(y_i,r)$ was added to $F$
in the coarsening phase. 
If $\disk(y_i,r)$ was not added to $F$, then it must have been intersected by some disk $\disk(y_{i'},r')$ that was added to $F$, such that $r' \geq r$. Clearly,
$x_j \in \disk(y_{i'},3r')$. Furthermore, $3r' \geq r \geq ||y^{\kappa}(x_j) - x_j||$. Thus, $x_j \in C_{i'}(3r')$, and $x_j$ is served by the output of $\oc(X', Y, \kappa, \alpha)$.

\subsection{Approximation Ratio}
Let the set of disks in an optimal outer cover be denoted by $OPT$.
We now show that the cost of the outer cover returned by 
$\oc(X', Y, \kappa, \alpha)$ is at most $3^{\alpha} \cdot \text{cost}(OPT)$. We begin
by lower bounding $\text{cost}(OPT)$ in terms of the $\beta_j$. We
have 

\begin{equation}
\label{eq:LPTight}
\text{cost}(OPT) \geq \sum_{\disk(y_i,r) \in OPT} \left( \sum_{x_j \in C_i(r)} \beta_j \right)
\geq \sum_{x_j \in X'} \beta_j.
\end{equation}

The first inequality follows because the $\beta_j$ satisfy (\ref{eq:dual_constraint}); the second is because each client in $X'$ is served by at least one disk
in  $OPT$, and the $\beta_j$ are non-negative.

Let $C$ denote the cost of the solution returned by $\oc(X', Y, \kappa, \alpha)$. We have
\begin{equation*}
C = 3^{\alpha} \cdot \text{cost}(F) = 3^{\alpha} \sum_{\disk(y_i,r) \in F} \left( \sum_{x_j \in
C_i(r)} \beta_j \right) \leq 3^{\alpha} \sum_{x_j \in X'} \beta_j \leq 3^{\alpha} 
\cdot \text{cost}(OPT).
\end{equation*} 
Here, the second equality is because each disk in $F$ is tight; since the disks
in $F$ are pairwise disjoint, each client $x_j \in X'$ is contained in at
most one disk in $F$, from which the next inequality follows; the final
inequality is due to Inequality~(\ref{eq:LPTight}).

Thus, we may conclude:

\begin{lemma}
\label{lem:oc}
The algorithm $\oc(X', Y, \kappa,\alpha)$ runs in polynomial time and 
returns an outer cover whose cost is at most  $3^{\alpha}$ times that of
an optimal outer cover.
\end{lemma}

\section{Computing a covering for the non-uniform MCMC problem}
\label{sec:algo}
With our algorithm for computing an outer cover in place, we now address
the non-uniform MCMC problem. Recall that the input is a client set $X$,
a server set $Y$, a coverage function $\kappa: X \rightarrow \Natural
\cup \{0\}$, and the constant $\alpha$.  

Given an assignment of radius $r_y$ to each $y \in Y$, we will say that
a point $x \in X$ is {\em $j$-covered} if at least $j$ disks cover it,
that is, 
\[ | \{ y \in Y \ | \ ||x-y|| \leq r_y \} | \geq j.\]
We will sometimes say that $x$ is
$\kappa$-covered to mean that it is $\kappa(x)$-covered. Similarly, if we have a assignment of radii to each 
$y \in Y$  such that for a set of points $P \subseteq X$, every point $x \in P$ is covered
by at least $\kappa(x)$ disks, we say that $P$ is $\kappa$-covered.

 \begin{algorithm*}[h]
 \caption{$\cover(X,Y,\kappa,\alpha)$}
 \label{alg:MC}
 \begin{algorithmic}[1]
   \IF {$\forall x \in X, \kappa(x) = 0$}
       \STATE Assign $r_y \leftarrow 0$ for each $y \in Y$, and return.
   \ENDIF
   \STATE Define $\kappa ' (x)$ as follows:
   \[ 
      \forall x \in X, \kappa '(x) = 
         \begin {cases}
            0, & \mbox{if } \kappa(x) = 0 \\
            \kappa (x) - 1, & \mbox{if } \kappa(x) > 0
         \end {cases}
   \]
   \STATE Recursively call $\cover(X,Y,\kappa',\alpha)$. 

   \STATE Let $X' \leftarrow \{x \in X \ | \ x \text{ is not } \kappa(x)\text{-covered } \}$
   \STATE Call the procedure $\oc(X',Y,\kappa,\alpha)$ to obtain an
          outer cover $\rho: Y \rightarrow \Real^+$.
   \STATE Let $Y' \leftarrow Y$.
   \STATE Let $\primary \leftarrow \varnothing$.
   \WHILE {$X' \neq \varnothing$}
   \STATE Choose $\overline{y} \in Y'$.
   \STATE $\primary \leftarrow \primary \cup \{\overline{y}\}$.
   \STATE Let $\xcluster_{\overline{y}} \leftarrow \varnothing$, $\ycluster_{\overline{y}} \leftarrow \varnothing$.
     \FORALL {$x' \in X'$}
     \IF {$x' \in \disk(\overline{y}, \rho_{\overline{y}})$ and $\rho_{\overline{y}} \geq ||x' - y ^{\kappa}(x')||$}
           \STATE $\xcluster_{\overline{y}} \leftarrow \xcluster_{\overline{y}} \cup \{x'\}$.
           \STATE $\ycluster_{\overline{y}} \leftarrow \ycluster_{\overline{y}} \cup \{y^1(x'), y^2(x'), \ldots, y^{\kappa}(x')\}.$
        \ENDIF
     \ENDFOR
     \STATE Let $\ycluster'_{\overline{y}} \subseteq \ycluster_{\overline{y}}$ be a set of at most
            four points such that 
            \[ \bigcap_{y \in \ycluster'_{\overline{y}}} \disk(y, r_y) = \bigcap_{y \in \ycluster_{\overline{y}}} \disk(y, r_y). \]
     \STATE For each $y \in \ycluster'_{\overline{y}}$, increase $r_y$ by the smallest
            amount that ensures $\xcluster_{\overline{y}} \subseteq \disk(y,r_y)$.
     \STATE Remove $\overline{y}$ from $Y'$ and remove from $X'$ any points $x$ that are ${\kappa(x)}$-covered.
   \ENDWHILE 
\end{algorithmic}
\end{algorithm*} 

Our algorithm $\cover(X,Y,\kappa,\alpha)$ for non-uniform MCMC computes an assignment of radius $r_y$ to each server $y \in Y$ such that each client $x \in X$ is $\kappa(x)$-covered.
This algorithm is recursive, and in the base case we 
have $\kappa(x) = 0$ for each $x \in X$. In the base case, the radius $r_y$ is assigned to
$0$ for each $y \in Y$. Otherwise, we define 
\[ \kappa'(x) = \max \{0, \kappa(x) - 1\}, \mbox{ for each } x \in X,\]
and recursively call
$\cover(X,Y,\kappa',\alpha)$ to compute an assignment that $\kappa'(x)$-covers each $x \in X$.
We then compute $X' \subseteq X$, the set of points that are not
$\kappa(x)$-covered. We compute an outer cover $\rho: Y \rightarrow \Real^+$ 
for $X'$ using the procedure $\oc(X',Y,\kappa,\alpha)$ described in 
Section~\ref{sec:outercover}. For any client $x \in X'$, the outer cover has
a disk $\disk(y, \rho_y)$ that serves it. That is, $x$ is contained in
$\disk(y, \rho_y)$ and $\rho_y \geq ||x - y^{\kappa}(x)||$.

The goal of the while-loop is to increase some of
the $r_y$ to ensure that each $x \in X'$, which is 
currently $(\kappa(x) - 1)$-covered, is also $\kappa(x)$-covered.
To do this, we iterate via the while loop over each disk $\disk(\overline{y}, \rho_{\overline{y}})$ returned by $\oc(X',Y,\kappa,\alpha)$. We add all
points in $X'$ that are served in the outer cover by $\disk(\overline{y}, \rho_{\overline{y}})$ to a set $\xcluster_{\overline{y}}$. That is, $\xcluster_{\overline{y}}$ consists of all $x' \in X'$ that are contained in  $\disk(\overline{y}, \rho_{\overline{y}})$ and $\rho_{\overline{y}} \geq ||x' - y ^{\kappa}(x')||$.  The set $\ycluster_{\overline{y}}$ contains, for
each $x \in \xcluster_{\overline{y}}$, the ${\kappa(x)}$ nearest neighbors of $x$ in $Y$. For purposes of analysis, we add $\overline{y}$ to a set $\primary$ as well.

Next, we identify a set $\ycluster'_{\overline{y}} \subseteq \ycluster_{\overline{y}}$ of at most
$4$ points such that 
 \[ \bigcap_{y \in \ycluster'_{\overline{y}}} \disk(y, r_y) = \bigcap_{y \in \ycluster_{\overline{y}}} \disk(y, r_y). \]
Why does such a $\ycluster'_{\overline{y}}$ exist? If, on the one hand, the intersection of
disks $\bigcap_{y \in \ycluster_{\overline{y}}} \disk(y, r_y)$ is empty, then Helly's Theorem tells us that there are three disks (or maybe even two) whose intersection is empty. On the other hand,
if the intersection $\bigcap_{y \in \ycluster_{\overline{y}}} \disk(y, r_y)$ is non-empty, then
it is a rectangle (as these are $\Linfty$ disks) and therefore equal to the intersection of four of the disks.

We enlarge the radius $r_y$ of each $y \in \ycluster'_{\overline{y}}$ by the minimum
amount needed to ensure that $\xcluster_{\overline{y}} \subseteq \disk(y,r_y).$ We
argue that after this each point in $\xcluster_{\overline{y}}$ is $\kappa$-covered. To
see why, consider any $x' \in \xcluster_{\overline{y}}$. Notice that $|\ycluster_{\overline{y}}| \geq {\kappa(x')}$, since
the $\kappa(x')$ nearest neighbors of $x'$ are included in $\ycluster_{\overline{y}}$. Thus before the enlargement,
$x'$ does not belong to $\bigcap_{y \in \ycluster_{\overline{y}}} \disk(y, r_y)$.
(Recall that no point in $\xcluster_{\overline{y}}$ was $\kappa$-covered.) Therefore,
$x'$ does not belong to $\bigcap_{y \in \ycluster'_{\overline{y}}} \disk(y, r_y)$.
It follows that there is at least one
$y \in \ycluster'_{\overline{y}}$ such that $\disk(y,r_y)$ did not contain $x'$ before
the enlargement. As a consequence of the enlargement, $\disk(y,r_y)$
does contain $x'$. Since $x'$ was $(\kappa(x') - 1)$-covered before the enlargement, it is now ${\kappa(x')}$-covered.

After increasing $r_y$ for $y \in \ycluster'_{\overline{y}}$ as stated, we discard
from $X'$ all points that are now $\kappa$-covered. The discarded set contains
$\xcluster_{\overline{y}}$ and possibly some other points in $X'$. We remove $\overline{y}$ from $Y'$. We go back and
iterate the while loop with the new $X'$ and $Y'$.

Since any point in $X'$ as computed in Line~5 is served by some disk in the
outer cover, it appears in $\xcluster_{\overline{y}}$ in some iteration of
the while loop (if it has not already been $\kappa$-covered serendipitously). 
At the end of that iteration of the while loop, it gets $\kappa$-covered. 
Thus, when $\cover(X,Y,\kappa,\alpha)$ terminates, each point $x \in X$ is
${\kappa(x)}$-covered.

\section{Approximation Ratio}
\label{sec:analysis}
In this section, we bound the ratio of the cost of the solution returned by
$\cover(X,Y,\kappa,\alpha)$ and the cost of the optimal solution. For this
purpose, the following lemma is central. It bounds the increase in cost
incurred by $\cover(X,Y,\kappa,\alpha)$ in going from a $\kappa'$-cover
to a $\kappa$-cover by the cost of the outer cover $\rho$ for $X'$.

\begin{lemma}
The increase in the objective function $\sum_{y \in Y}r_y^{\alpha}$ from the
time $\cover(X,Y,\kappa',\alpha)$ completes to the time $\cover(X,Y,\kappa,\alpha)$ completes
is $4 \cdot 3^{\alpha} \cdot \sum_{y \in Y} \rho_y^{\alpha}$.
\label{lem:increase}
\end{lemma}

\begin{proof}
    Let us fix an $\overline{y} \in \primary$, and focus on the iteration when $\overline{y}$ was
added to $\primary$. Notice that there is exactly one such iteration, 
since $\overline{y}$ is removed from $Y'$ in the iteration it gets added to
$\primary$.

We will bound the increase in cost during this iteration. For this, we need
two claims.

\begin{claim}
    For any $x' \in \xcluster_{\overline{y}}$, we have 
    \[ ||\overline{y} - x'|| \leq \rho_{\overline{y}} \]
\end{claim}  

\begin{proof}
Recall that $x'$ is in $\xcluster_{\overline{y}}$ because $x' \in \disk(\overline{y}, \rho_{\overline{y}})$. 
\end{proof}

\begin{claim}
    For any $y' \in \ycluster_{\overline{y}}$, we have 
    \[||y' - \overline{y}|| \leq 2*\rho_{\overline{y}} \]
\end{claim}

\begin{proof}
Let $y'$ be added to $\ycluster_{\overline{y}}$ when $x' \in X'$ was added to $\xcluster_{\overline{y}}$.
Hence
\begin{align*}
   ||y' - x' || &\leq ||x' - y^{\kappa}(x')|| \\
                &\leq \rho_{\overline{y}},
\end{align*}
since $\disk(\overline{y}, \rho_{\overline{y}})$ serves $x'$ in the outer
cover (line 14 of Algorithm~\ref{alg:MC}). 
Also, since $x' \in \disk(\overline{y}, \rho_{\overline{y}})$, $||x' - \overline{y}|| \leq \rho_{\overline{y}}$.
Therefore,
\begin{align*}
   ||y' - \overline{y}|| &\leq ||y' - x'|| + ||x' - \overline{y}|| \\
                         &\leq \rho_{\overline{y}} + \rho_{\overline{y}} \\
                         &= 2\rho_{\overline{y}}
\end{align*}
\end{proof}

Fix a $y \in \ycluster'_{\overline{y}}$. If $r_y$ was increased in this iteration,
it now equals $||y-x'||$ for some $x' \in \xcluster_{\overline{y}}$. By the above
two claims, 
\begin{align*}
   ||y-x'|| &\leq ||y - \overline{y}|| + ||\overline{y} - x'||\\
            &\leq 3*\rho_{\overline{y}}
\end{align*}

Thus the increase in $r_y^{\alpha}$ is at most $3^{\alpha} (\rho_{\overline{y}})^{\alpha}$. Since $r_y$ is increased in this iteration only for $y \in \ycluster'_{\overline{y}}$, and $|\ycluster'_{\overline{y}}| \leq 4$, the increase in the objective
function $\sum_{y \in Y} r_y^{\alpha}$ (in the iteration
of the while loop under consideration) is 
at most $4 \cdot 3^{\alpha} \cdot (\rho_{\overline{y}})^{\alpha}$.

We conclude that the increase in  $\sum_{y \in Y} r_y^{\alpha}$ over all the
iterations of the while loop is at most
\[ 4 \cdot 3^{\alpha} \cdot \sum_{\overline{y} \in \primary} (\rho_{\overline{y}})^{\alpha} 
= 4 \cdot 3^{\alpha} \cdot \sum_{y \in Y} \rho_y^{\alpha}.\]
\end{proof}

We can now bound the approximation ratio of the algorithm.

\begin{lemma}
 Let $r':Y \rightarrow \Real^+$ be any assignment of radii to the points in $Y$ 
    under which each point $x \in X$ is ${\kappa(x)}$-covered. Then the cost of the output
of $\cover(X,Y,\kappa,\alpha)$ is at most $c * \sum_{y \in Y} {r'_y}^{\alpha}$, where $c = 4 \cdot 27^{\alpha}$.
\label{lem:constant}
\end{lemma}

\begin{proof}
Our proof is by induction on $\max_{x \in X} \kappa(x)$. For the base case, where $\kappa(x) = 0$ for each $x \in X$, the claim in
the theorem clearly holds. 

Let $D = \{\disk(y, r'_y) \ | \ y \in Y\}$ be the set of disks corresponding
to the assignment $r'$. Our proof strategy is to show that there is a subset $D_{\kappa} \subseteq D$ such that

\begin{enumerate}
   \item The cost increase incurred by $\cover(X,Y,\kappa,\alpha)$ in going from
the $\kappa'$-cover to the $\kappa$-cover is at most $c$ times cost of the disks in $D_k$. (Recall that the \textit{cost} of a set of  
disks is the sum of the $\alpha$-th powers of the radii of the disks.)
   \item The set of disks, $D \setminus D_{\kappa}$, $\kappa'(x)$-covers any point $x \in X$. 
\end{enumerate}
 
By the induction hypothesis, the cost of the $\kappa'$-cover computed by
$\cover(X,Y,\kappa,\alpha')$ is at most $c$ times the cost of the disks in $D \setminus D_{\kappa}$. As
the increase in cost incurred by $\cover(X,Y,\kappa,\alpha)$ in turning the $\kappa'$-cover
to a $\kappa$-cover is at most $c$ times the cost of the disks in $D_{\kappa}$, the theorem follows.

We now describe how $D_k$ is computed, and then establish that it has the above
two properties. For each $x' \in X'$, let $\largest(x')$
be the largest disk from $D$ that contains $x'$. Since $x'$ is $\kappa(x')$-covered by $D$, we note that the radius of $\largest(x')$ is at least $||x' - y^{\kappa}(x')||$. Let 
\[ D'_{\kappa} = \{ \largest(x') \  | \ x' \in X' \}.\]

Sort the disks in $D'_{\kappa}$ by decreasing (non-increasing) radii. Let $B
\leftarrow \varnothing$ initially. For each disk $d \in D'_{\kappa}$ in the sorted
order, performing the following operation: add $d$ to $B$ if $d$ does
not intersect any disk already in $B$.

Let $D_{\kappa}$ be the set $B$ at the end of this computation. Since no two
disks in $D_{\kappa}$ intersect, and $D$ $\kappa$-covers any point in $X$, it follows
that $D \setminus D_{\kappa}$ $\kappa '$-covers any point in $X$. This establishes
Property 2 of $D_{\kappa}$.

We now turn to Property 1. For this, consider $L_{\kappa}$, the set of disks obtained by increasing the
radius of each disk in $D_{\kappa}$ by a factor of 3. We argue that $L_{\kappa}$ is
an outer cover for $X'$. Fix any $x' \in X'$.
\begin{enumerate}
\item If $\largest(x') \in D_{\kappa}$, then the corresponding disk in 
$L_{\kappa}$ contains $x'$ and has radius at least $||x' - y^{\kappa}(x')||$.
\item If $\largest(x') \not\in D_{\kappa}$, then there is an even larger disk in 
$D_{\kappa}$ that intersects  $\largest(x')$. The corresponding disk in $ L_{\kappa}$ contains $x'$ and has radius at least $||x' - y^{\kappa}(x')||$.
\end{enumerate}

Since $L_{\kappa}$ is
an outer cover for $X'$, and the procedure 
$\oc(X',Y,\kappa,\alpha)$ returns a $3^{\alpha}$ approximation
to the optimal outer cover, we infer that
\[ \sum_{y \in Y} \rho_y^{\alpha} \leq 3^{\alpha} \cdot \text{cost}(L_{\kappa})
\leq 9^{\alpha} \cdot \text{cost}(D_{\kappa}).\]

Thus the cost increase incurred by $\cover(X,Y,\kappa,\alpha)$ in going from
the $\kappa'$-cover to the $\kappa$-cover is, by Lemma~\ref{lem:increase}, at most 
\[ 4 \cdot 3^{\alpha} \cdot \sum_{y \in Y} \rho_y^{\alpha} \leq 
   4 \cdot 27^{\alpha} \cdot \text{cost}(D_{\kappa}) = c \cdot \text{cost}(D_{\kappa}).\]

This establishes Property 1, and completes the proof of the lemma. 
\end{proof}

We conclude with a statement of the main result of this article. In this statement, cost refers to $l_2$ rather than $l_{\infty}$ disks. Since (a) an $l_2$ disk of radius $r$ is contained in the corresponding $l_{\infty}$ disk of radius 
$r$, and (b) an $l_{\infty}$ disk of radius $r$ is contained in an $l_2$ disk
of radius $\sqrt{2}r$, the approximation guarantee is increased by
$(\sqrt{2})^{\alpha})$ when compared to Lemma~\ref{lem:constant}. 

\begin{theorem}
\label{thm:main}
Given point sets $X$ and $Y$ in the plane, a coverage function $\kappa: X \rightarrow \{0,1,2,\ldots, |Y|\}$, and $\alpha \geq 1$, the algorithm $\cover(X, Y, \kappa, \alpha)$ runs in polynomial time and computes a $\kappa$-cover of $X$ with cost at most $4 \cdot (27 \sqrt{2})^{\alpha}$ times that of the optimal $\kappa$-cover.
\end{theorem}

\section{Concluding Remarks}
\label{sec:conclusions}
Our result generalizes to the setting where $X$ and $Y$ are points in 
$\Real^d$, where $d$ is any constant. The approximation guarantee is now
$(2d) \cdot (27 \sqrt{d})^{\alpha}$. To explain, the intersection of a
finite family of $l_{\infty}$ balls equals the intersection of a sub-family
of at most $2d$ balls. That is why the $4$ in the approximation guarantee
of Theorem~\ref{thm:main} becomes $2d$. In the transition from $l_2$ to
$l_{\infty}$ balls in $\Real^d$, we lose a factor of $(\sqrt{d})^{\alpha}$. 

This generalization naturally leads to the next question -- what can we say
when $X$ and $Y$ are points in an arbitrary metric space? Our approach
confronts a significant conceptual obstacle here, since one can easily 
construct examples in which the cost of going from a $(k-1)$-cover to
a $k$-cover (for the uniform MCMC) cannot be bounded by a constant times the
cost of an optimal outer cover. Thus, new ideas seem to be needed for
obtaining an $O(1)$ approximation for this problem. The work of \cite{Yehuda2013Note} gives the best known guarantee
of $O(k)$. For the non-uniform version, their approximation guarantee is
$O(\max \{\kappa(x) \mid x \in X\})$.

\end{document}